\documentclass[twoside]{article}

\usepackage[accepted]{aistats2021_arxiv}
\usepackage{graphicx}
\usepackage{subcaption}
\usepackage[colorlinks=true,citecolor=black]{hyperref}
\usepackage{amsmath,amsfonts,amsthm,amssymb}
\usepackage{bbm}
\usepackage{bm}
\usepackage{xcolor}

\newtheorem{defn}{Definition}
\newtheorem{example}{Example}

\newtheorem{prop}{Proposition}
\newtheorem{cor}{Corollary}

\usepackage{prettyref}
\newcommand{\rref}[2][]{\prettyref{#2}}
\newrefformat{model}{Model\,\ref{#1}}
\newrefformat{listing}{Listing\,\ref{#1}}
\newrefformat{alg}{Algorithm\,\ref{#1}}
\newrefformat{line}{line\,\ref{#1}}
\newrefformat{sec}{Section\,\ref{#1}}
\newrefformat{subsec}{Subsection\,\ref{#1}}
\newrefformat{section}{Section\,\ref{#1}}
\newrefformat{appendix}{Appendix\,\ref{#1}}
\newrefformat{app}{Appendix\,\ref{#1}}
\newrefformat{def}{Definition\,\ref{#1}}
\newrefformat{defn}{Definition\,\ref{#1}}
\newrefformat{thm}{Theorem\,\ref{#1}}
\newrefformat{ax}{\ref{#1}}
\newrefformat{prop}{Prop.\,\ref{#1}}
\newrefformat{lemma}{Lemma\,\ref{#1}}
\newrefformat{cor}{Corollary\,\ref{#1}}
\newrefformat{corollary}{Corollary\,\ref{#1}}
\newrefformat{ex}{Example\,\ref{#1}}
\newrefformat{tab}{Table\,\ref{#1}}
\newrefformat{fig}{Fig.\,\ref{#1}}
\newrefformat{eqn}{Equation~(\ref{#1})}
\newrefformat{problem}{Problem\,\ref{#1}}

\newcommand{\ci}{\mathrel{\perp\mspace{-10mu}\perp}}

\newcommand{\inv}{{-1}}

\newcommand{\cD}{\mathcal{D}}
\newcommand{\cE}{\mathcal{E}}
\newcommand{\cF}{\mathcal{F}}
\newcommand{\cG}{\mathcal{G}}

\newcommand{\cM}{\mathcal{M}}

\newcommand{\cR}{\mathcal{R}}

\newcommand{\cV}{\mathcal{V}}

\newcommand{\bbR}{\mathbb{R}}

\newcommand{\bbP}{\mathbb{P}}

\DeclareMathOperator{\pa}{pa}

\DeclareMathOperator{\nbr}{ne}

\DeclareMathOperator{\descendants}{de}

\DeclareMathOperator{\pre}{pre}

\newcommand{\hatrho}{\hat{\rho}}
\newcommand{\hatz}{\hat{z}}
\newcommand{\hatTheta}{\hat{\Theta}}
\newcommand{\hatSigma}{\hat{\Sigma}}
\newcommand{\hatK}{\hat{K}}
\newcommand{\maximal}{\textrm{Maximal}}
\newcommand{\append}{\texttt{append}}
\newcommand{\Path}{\texttt{path}}
\newcommand{\paths}{\texttt{paths}}
\newcommand{\newpaths}{\texttt{new\_paths}}

\newcommand{\RFDStep}{\texttt{RFDStep}}
\newcommand{\RFD}{\texttt{RFD}}

\usepackage[round]{natbib}

\usepackage{algorithm}
\usepackage{algpseudocode}
\usepackage[symbol]{footmisc}

\begin{document}

%

%

\twocolumn[

\aistatstitle{Efficient Permutation Discovery in Causal DAGs}

\renewcommand{\thefootnote}{\fnsymbol{footnote}}
\aistatsauthor{ Chandler Squires\textsuperscript{*} \And Joshua Amaniampong\textsuperscript{*} \And  Caroline Uhler }

\renewcommand{\thefootnote}{\arabic{footnote}}

\aistatsaddress{ MIT \And  MIT \And MIT } ]

\begin{abstract}
    The problem of learning a directed acyclic graph (DAG) up to Markov equivalence is equivalent to the problem of finding a permutation of the variables that induces the sparsest graph. 
    Without additional assumptions, this task is known to be NP-hard. 
    Building on the minimum degree algorithm for sparse Cholesky decomposition, but utilizing DAG-specific problem structure, we introduce an efficient algorithm for finding such sparse permutations.
    %
    %
    We show that on jointly Gaussian distributions, our method with depth $w$ runs in $O(p^{w+3})$ time. 
    We compare our method with $w = 1$ to algorithms for finding sparse elimination orderings of undirected graphs, and show that taking advantage of DAG-specific problem structure leads to a significant improvement in the discovered permutation. 
    We also compare our algorithm to \textit{provably consistent} causal structure learning algorithms, such as the PC algorithm, GES, and GSP, and show that our method achieves comparable performance with a shorter runtime. Thus, our method can be used on its own for causal structure discovery.
    Finally, we show that there exist dense graphs on which our method achieves almost perfect performance, so that unlike most existing causal structure learning algorithms, the situations in which our algorithm achieves both good performance and good runtime are not limited to sparse graphs.
\end{abstract}

\section{Introduction}\label{sec:introduction}

The discovery of causal structure, represented by a directed acyclic graph (DAG), from data has received much attention over the past two decades \citep{spirtes2000causation,chickering2002optimal,shimizu2006linear,hauser2012characterization,peters2014causal,solus2017consistency}, due to the ability of causal models to answer questions about the effect of hypothetical interventions, such as ``how will a new treatment affect a patient's diabetes", or ``how will a new housing law affect rent prices"? Methods for causal structure discovery exploit a number of patterns in the data, including (1) conditional independencies, as in the PC algorithm \citep{spirtes2000causation}, (2) asymmetries arising from nonlinearities, as in the LiNGAM algorithm \citep{shimizu2006linear}, or (3) data likelihood, as in GES \citep{chickering2002optimal}. While some of these causal discovery algorithms are \textit{provably consistent}, in the sense that given infinite data they converge to the correct causal model (or one that is \textit{equivalent} to it, consistency comes at a high computational price; e.g.~the complexity of PC and GES grows exponentially in the maximum indegree of the graph, so that these algorithms are infeasible to run on large, dense graphs. Since the problem of causal structure discovery is known to be NP-hard \citep{chickering2004large}, this motivates the development of alternative, \textit{approximate} methods, which scale to dense graphs while still providing insights about the true causal graph.

\textbf{Finding causal orderings via sparsity.} A number of recent algorithms \citep{peters2014causal, raskutti2018learning, solus2017consistency,yang2018characterizing,Wang_NIPS_2017,squires2019permutation} have utilized the fact that, given the true causal ordering between variables (i.e., an ordering that is consistent with the topological ordering of the causal DAG), the causal graph can be easily recovered from conditional independence statements. Motivated by this connection, we here develop a method for inferring the causal ordering by exploiting the relationships between conditional independencies in a novel way. Our new method, which we call \emph{Removal-Fill-Degree (RFD)}, can be seen as an extension of the Minimum-Degree (MD) algorithm for finding sparse elimination orderings of undirected graphs \citep{rose1972graph,heggernes2001computational}. MD and modifications thereof have been well-studied for the purpose of finding sparse \textit{Cholesky decompositions}, which have a wide array of applications in numerical linear algebra \citep{rothberg1994efficient}. While these algorithms have focused on undirected graphs, in this paper, we introduce a number of modifications that are natural when extending to DAGs.

\textbf{Organization of the paper.}
In \rref{sec:background}, we review background on graphical models, as well as related work on DAG structure learning and algorithms for finding sparse Cholesky decompositions. 
In \rref{sec:theory}, we introduce new concepts and theoretical results which motivate our method. 
In \rref{sec:methods}, we introduce our \RFD~algorithm for finding permutations which induce sparse causal graphs, and evaluate its runtime.
We also describe a construction for a dense graph on which the \RFD~algorithm performs well, demonstrating that the computational efficiency of our algorithm does not limit it to perform well only on sparse graphs.
In \rref{sec:empirical}, we compare the \RFD~method to other methods on the tasks of causal structure discovery from (1) noiseless data and (2) noisy data.

\section{Background and Related Work}\label{sec:background}

\textbf{UGs, DAGs and IMAPs.} Given a graph over nodes $[p] := \{ 1, \ldots p \}$, we associate to each node $i$ a random variable $X_i$. In an undirected graph (UG) $\cG$, we write $i \ci_\cG j \mid S$ whenever two nodes $i$ and $j$ are \textit{separated} given $S$. Similarly, in a directed acyclic graph (DAG) $\cD$, we write $i \ci_\cD j \mid S$ whenever two nodes $i$ and $j$ are \textit{d-separated} given $S$ \citep{koller2009probabilistic}. We say that a distribution $\bbP$ is \emph{Markov} to a UG (DAG) if whenever $i$ and $j$ are (d-)separated given $S$ in the graph, we have $X_i \ci_\bbP X_j \mid X_S$. If $i$ and $j$ are not (d-) separated given $S$, we call them (d-)connected given $S$, and write $i \not\ci_\cD j \mid S$. See \cite{koller2009probabilistic} for a review of separation and connection statement in graphical models. We use $\pa_\cD(i)$, $\descendants_\cD(i)$ to refer to the \textit{parents} and \textit{descendants} of $i$ in a directed graph $\cD$. We use $\nbr_\cG(i)$ to refer to the neighbors of $i$ in an undirected graph; if $i$ and $j$ are neighbors in $\cG$ then we write $i \sim_\cG j$. We take for all of these terms the standard definitions from \cite{lauritzen1996graphical}. Given a path with consecutive nodes $(i, j, k)$, there is a \textit{collider} on $j$ in the path if $i \to j \leftarrow k$.

A graph $\cG'$ is an \textit{independence map (IMAP)} of a graph $\cG$ if every conditional independence statement entailed by $\cG'$ is also entailed by $\cG$. We denote an \emph{induced subgraph} of $\cG$ on the vertices $V$ by $\cG[V] = \{ i - j \mid i - j \in \cG, i, j \in V \}$. An IMAP $\cG'$ of $\cG$ is \textit{minimal} if no induced subgraph of $\cG'$ is an IMAP of $\cG$. Two DAGs $\cD$ and $\cD'$ are \emph{Markov equivalent} when they entail the same set of d-separation statements. Given a DAG $\cD$, for any permutation $\pi$ of $V(\cD)$, the graph
$$
\cD_\pi = \{ i \rightarrow j \mid i <_\pi j, i \not\ci_\cD j \mid \pre_\pi(j) \}
$$
is a minimal IMAP of $\cD$ \citep{verma1990causal}, where $\pre_\pi(j) = \{ k \mid k <_\pi j \}$.

\textbf{Ordering-Based Algorithms.} \cite{raskutti2018learning} established a fundamental connection between causal structure learning and the problem of finding permutations which induce sparse minimal IMAPs.
Based on this connection, \cite{solus2017consistency} introduced the \textit{Greedy Sparsest Permutation} (GSP) algorithm, which performs a greedy search over the space of permutations, searching for a graph with the minimum number of edges.
In order to establish high-dimensional consistency guarantees, the authors relied on a good choice of an initial permutation in order to limit the steps (and thus, the number of hypothesis tests) required to find the optimal permutation.

\textbf{Finding Sparse Elimination Orderings.}
As a heuristic for discovering a good initial permutation, \cite{solus2017consistency} proposed the use of the \textit{Minimum-Degree} (\textit{MD}) algorithm. 
The MD algorithm was designed for the problem of sparse Cholesky decomposition, i.e., finding a permutation of a matrix such that the lower diagonal component in its Cholesky decomposition is sparse.
Given an undirected graph and an order of the nodes, \textit{vertex elimination} is the process of iteratively removing each node in the order and connecting its neighbors.
Initial work on sparse Cholesky decompositions \citep{rose1972graph} provided a connection between finding sparse Cholesky decompositions and finding elimination orderings that introduce few edges.
Further work on this problem included improvements to time and space complexity, versions removing multiple nodes at a time \citep{liu1985modification}, and versions using an approximation of the degree \citep{amestoy1996approximate}.
See \cite{george1989evolution} or \cite{heggernes2001computational} for an extensive survey.

Since the MD algorithm and its current extensions are designed for undirected graphs, these methods fail to exploit patterns that are helpful for discovering topological orderings of DAGs, which we describe in this paper. 
Moreover, these methods assume that the original matrix is sparse, i.e., has entries exactly equal to zero.
However, in many applications, including causal structure learning, the matrix may be the result of some noisy process which only induces entries that are \textit{approximately} equal to zero. 
We provide details on how to efficiently incorporate conditional independence testing into methods for finding sparse elimination orderings.

\begin{figure}[t]
    \centering
    \includegraphics[width=\linewidth]{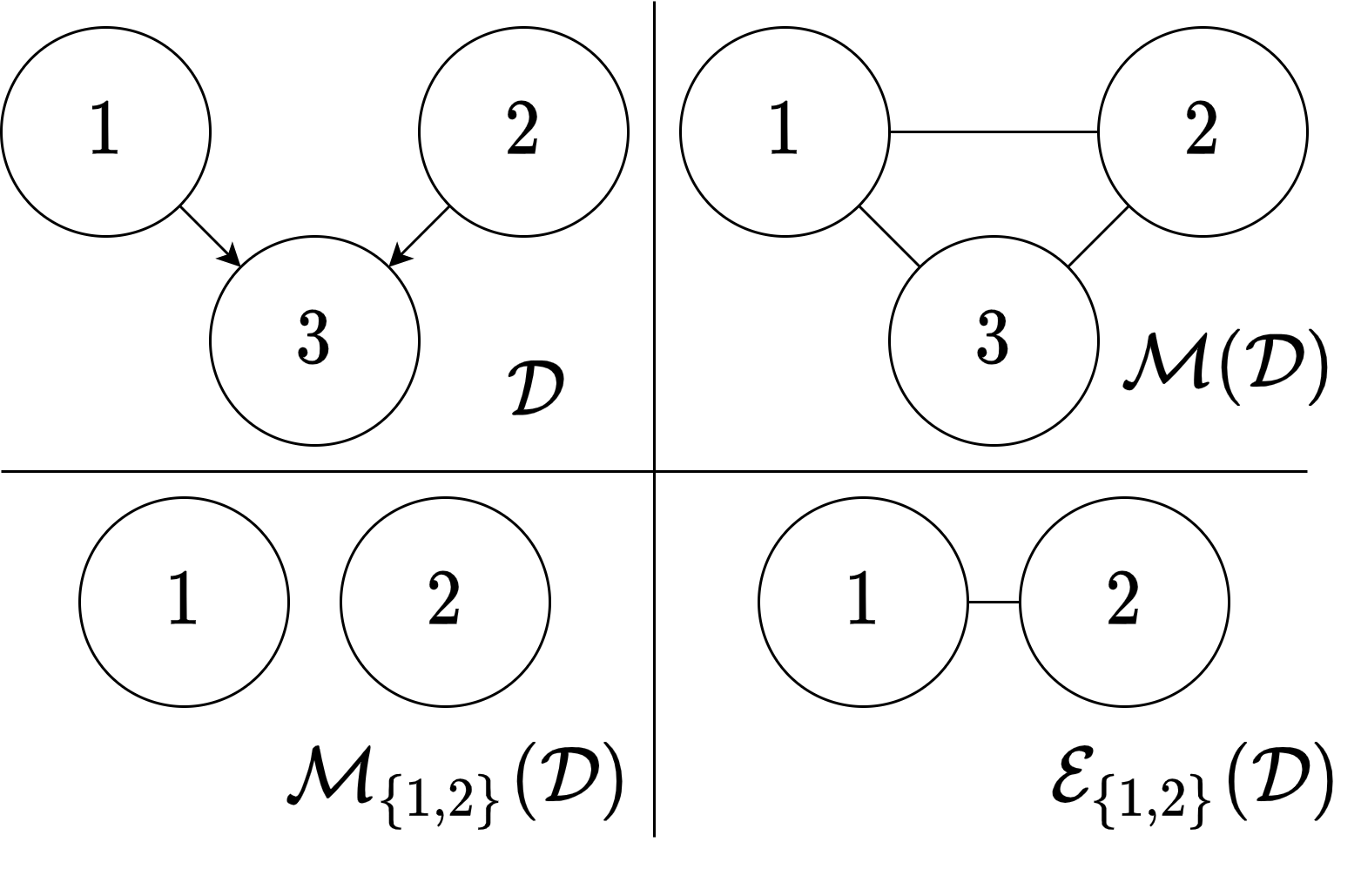}
    \caption{A DAG $\cD$, its moral graph $\cM(\cD)$, a moral subgraph $\cM_{\{1,2\}}(\cD)$, and elimination graph $\cE_{\{1,2\}}(\cD)$.}
    \label{fig:moral-subgraph}
\end{figure}

\section{Theoretical Results}\label{sec:theory}

The interplay between DAGs and undirected graphs will be central to our algorithm. Recall that given a DAG $\cD$, its \textit{moral graph} $\cM(\cD)$ is the unique undirected minimal IMAP of $\cD$ \citep{koller2009probabilistic}. We extend this concept to arbitrary subsets $\cV$ of the vertices $[p]=\{1,\dots , p\}$ of a DAG $\cD$:

\begin{figure*}
    \begin{subfigure}{0.16\linewidth}
        \centering
        \includegraphics[width=.9\textwidth]{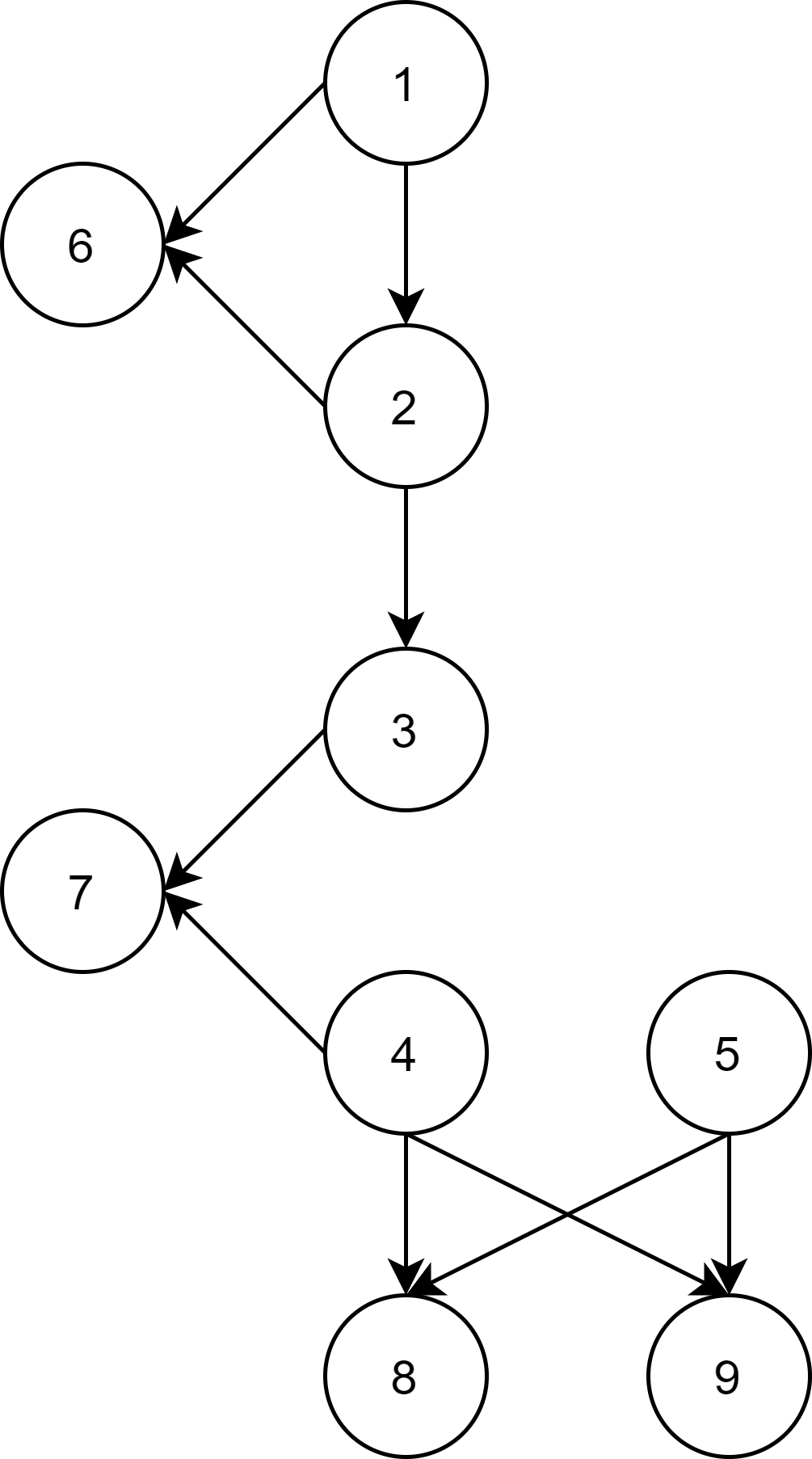}
        \caption{$\cD$}
        \label{fig:dag}
    \end{subfigure}
    \begin{subfigure}{0.16\linewidth}
        \centering
        \includegraphics[width=.9\textwidth]{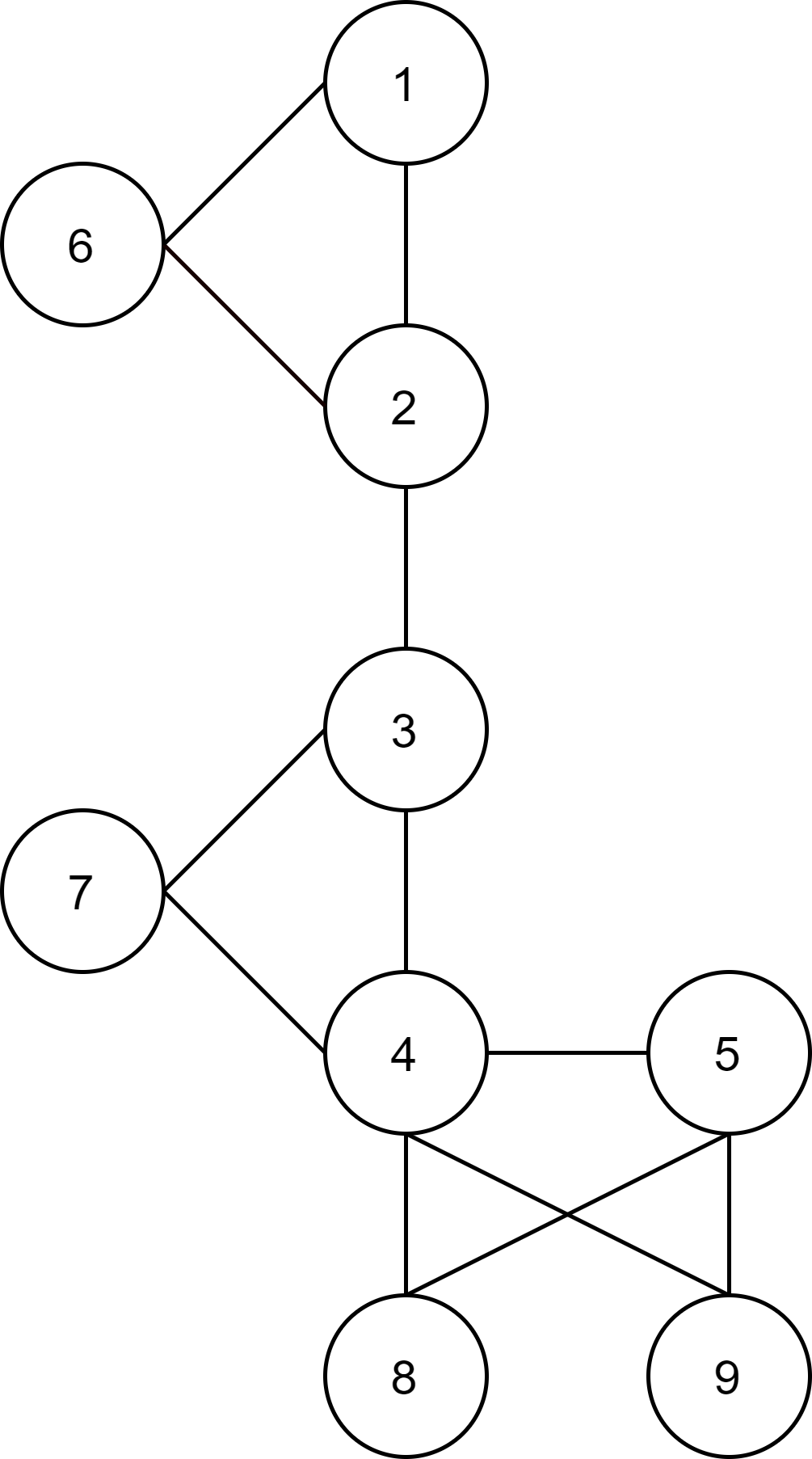}
        \caption{$\cM(\cD)$}
        \label{fig:mug}
    \end{subfigure}
    \begin{subfigure}{0.16\linewidth}
        \centering
        \includegraphics[width=.9\textwidth]{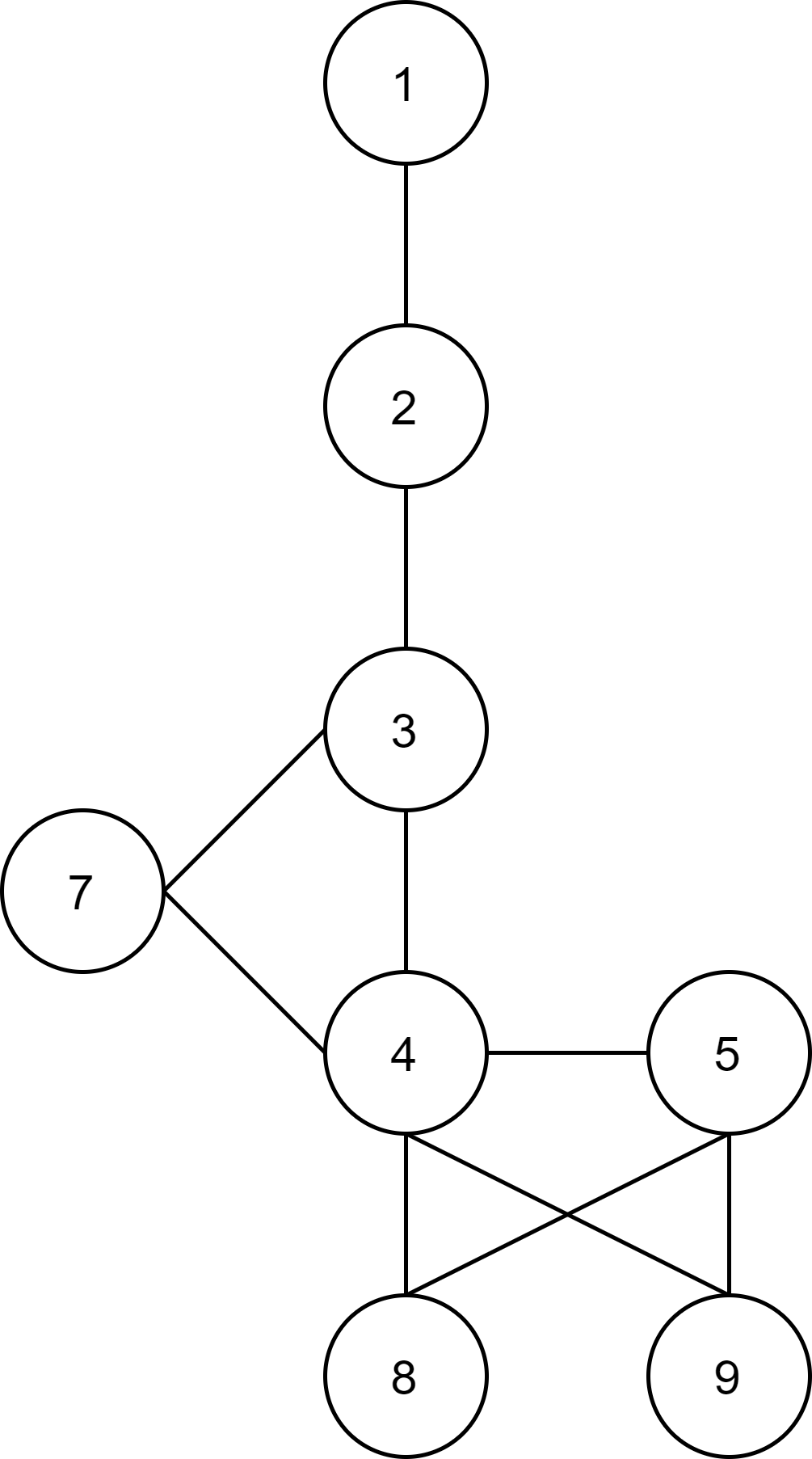}
        \caption{$\cM_{[9] \setminus 6}(\cD)$}
        \label{fig:mug6}
    \end{subfigure}
    \begin{subfigure}{0.16\linewidth}
        \centering
        \includegraphics[width=.9\textwidth]{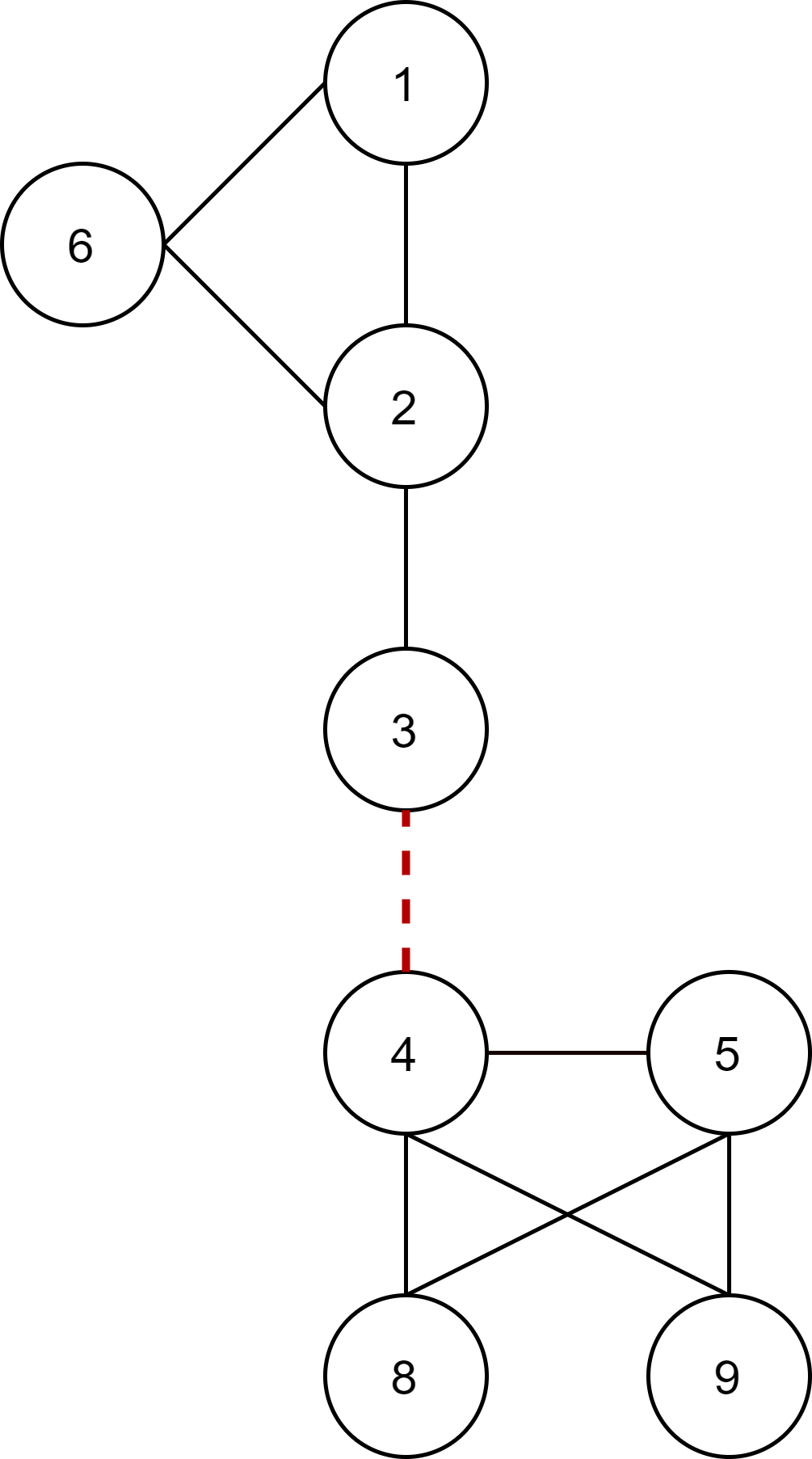}
        \caption{$\cM_{[9] \setminus 7}(\cD)$}
        \label{fig:mug7}
    \end{subfigure}
    \begin{subfigure}{0.16\linewidth}
        \centering
        \includegraphics[width=.9\textwidth]{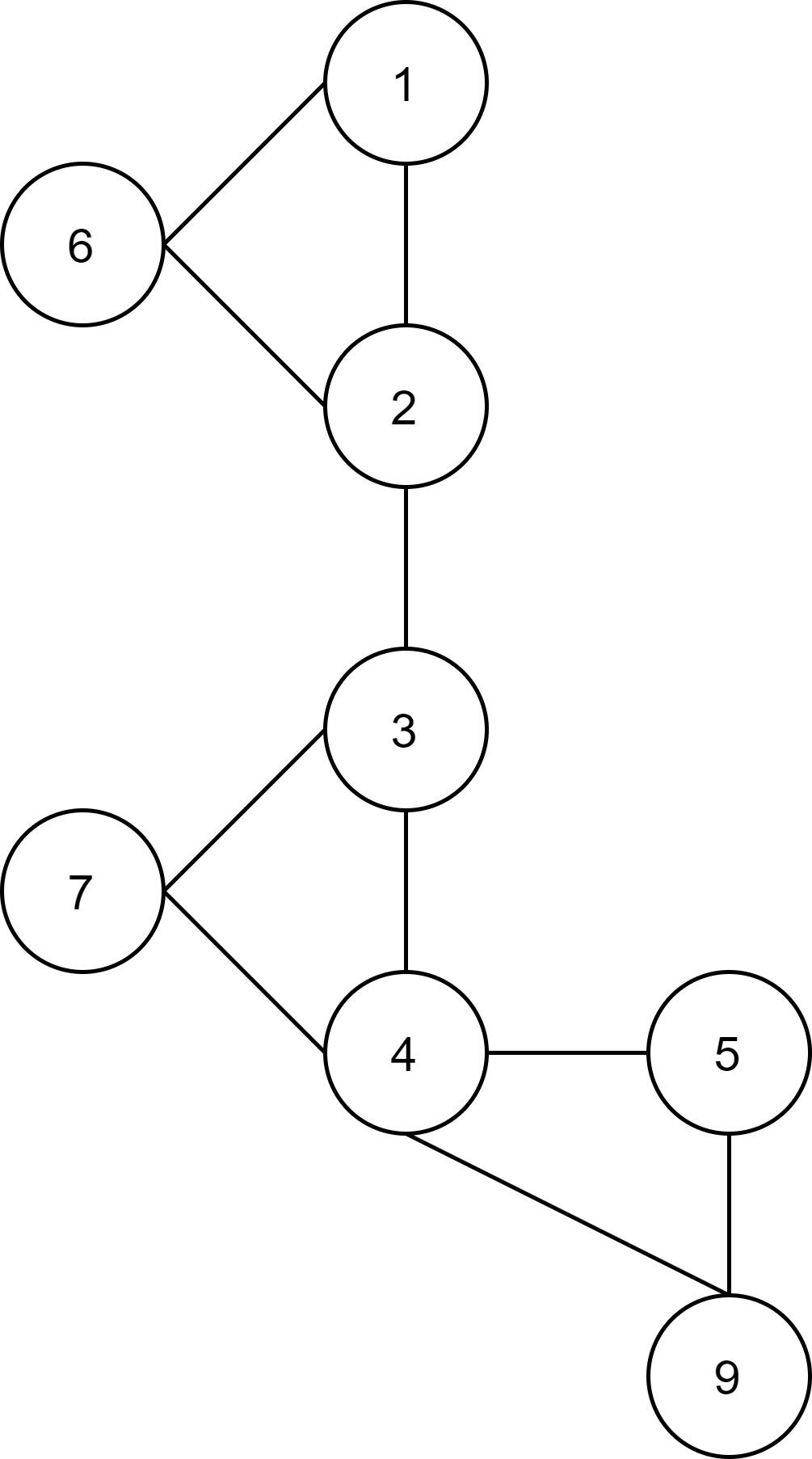}
        \caption{$\cM_{[9] \setminus 8}(\cD)$}
        \label{fig:mug8}
    \end{subfigure}
    \begin{subfigure}{0.16\linewidth}
        \centering
        \includegraphics[width=.9\textwidth]{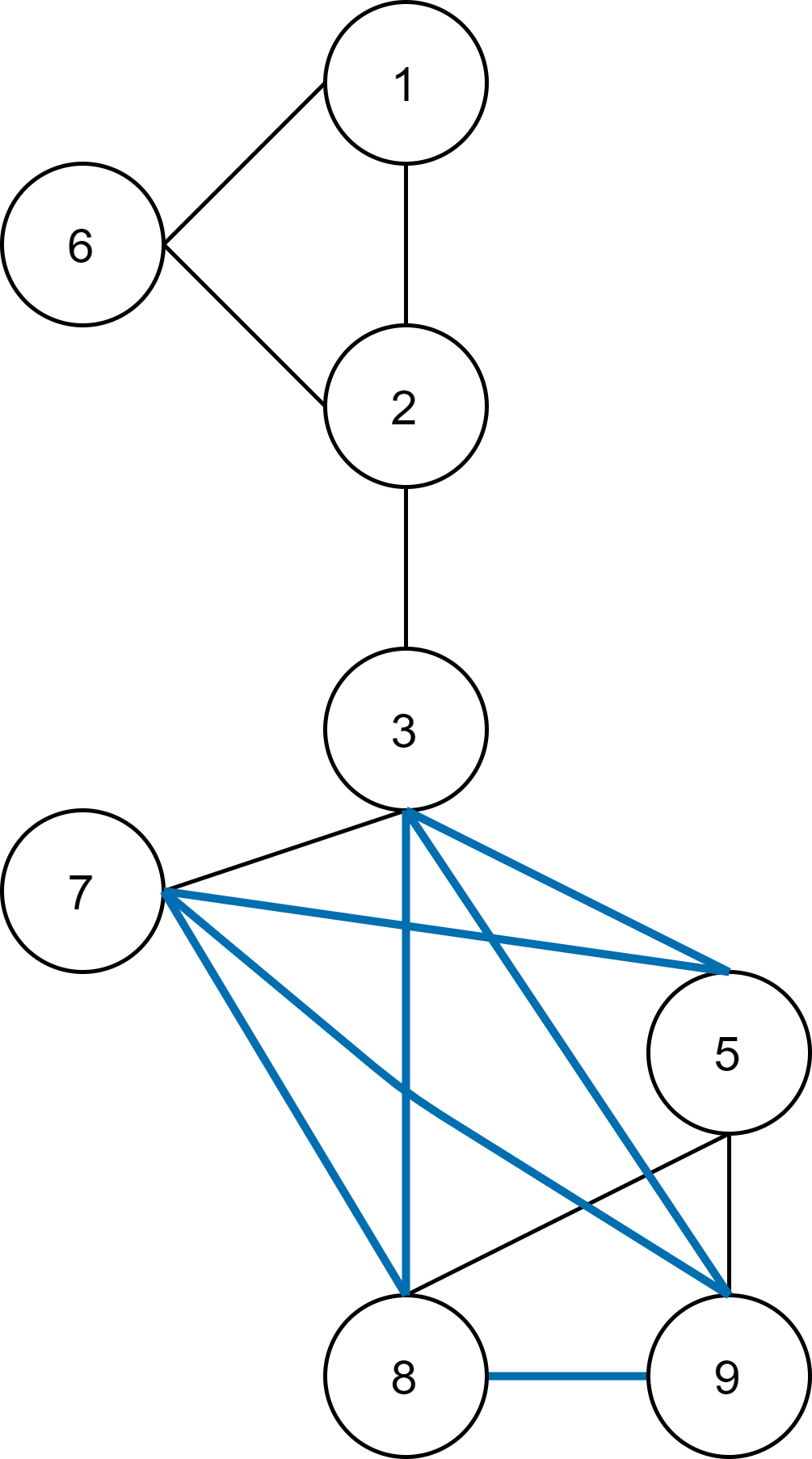}
        \caption{$\cM_{[9] \setminus 4}(\cD)$}
        \label{fig:mug4}
    \end{subfigure}
    
    \caption{A DAG $\cD$ and several of its moral subgraphs. Blue lines indicate fill edges, and dotted red lines indicate removed edges.}
    \label{fig:mug-examples}
\end{figure*}

\begin{defn} The \emph{moral subgraph} $\cM_{\cV}(\cD)$ of a DAG $\cD = \left([p],\cE\right)$ over vertices $\cV\subset [p]$ is the undirected graph  with vertex set $\cV$ and edge set
\[ 
\{ i - j \mid i, j \in \cV, i \not \ci_\cD j \mid \cV \setminus \{ i, j \} \}.
\]
\end{defn}

As with the moral graph, the moral subgraph is the unique undirected minimal IMAP of the marginal distribution $\bbP_\cV$. 

\textbf{The elimination graph.} The moral subgraph is closely related to the \textit{elimination graph} that has previously been studied in the sparse elimination ordering literature. 
The elimination graph for a given elimination ordering is defined by successively removing each node in the ordering and connecting its neighbors. 
Equivalently, the elimination graph can be described through the moral graph as follows: 
%
$\cV$ is $$\cE_\cV(\cD) = \{ i - j \mid i \not\ci_{\cM(\cD)} \cV \setminus \{ i, j \} \}.$$

Figure~\ref{fig:moral-subgraph} illustrates that the moral subgraph and elimination graph do not necessarily coincide; a notable exception is the case of chordal graphs when following a \textit{perfect elimination ordering} (see e.g.~\cite{vandenberghe2015chordal} for an excellent overview).

\textbf{Fill and removal edges.} Given a moral subgraph $\cM_\cV(\cD)$, we consider the effect of marginalizing out a node $k$. Since we will frequently add and remove single elements from sets, we let $\cV \setminus k := \cV \setminus \{ k \}$ and $\cV \cup k := \cV \cup \{ k \}$ for short. Marginalizing $k$ from $\cV$ results in the new moral subgraph $\cM_{\cV \setminus k}(\cD)$. Unlike in the elimination algorithm, removing a vertex does not only add edges, but may result in the removal of some edges. The edges removed and added after removing node $k$ are captured in the following definitions:

\begin{defn}
The \emph{removed edge set} of vertex $k$, over vertices $\cV$, is
\[
\cR_\cD(\cV, k) = \cM_\cV(\cD)[\cV \setminus k] \setminus \cM_{\cV \setminus k}(\cD).
\]
The \emph{removal score} of  $k$ on $\cV$ is
$R_\cD(\cV, k) = |\cR_\cD(\cV, k)|$.
\end{defn}

\begin{defn}
The \emph{fill edge set} of vertex $k$, over vertices $\cV$, is
\[
\cF_\cD(\cV, k) = \cM_{\cV \setminus k}(\cD) \setminus \cM_{\cV}(\cD)[\cV \setminus k ].
\]
The \emph{fill score} of $k$ over $\cV$ is
$
F_\cD(\cV, k) = |\cF_\cD(\cV, k)|.
$
\end{defn}

The following example demonstrates these definitions.

\begin{example}\label{ex:intuition}
\rref{fig:mug-examples} shows a DAG $\cD$ (\rref{fig:dag}), its moral graph $\cM(\cD)$ (\rref{fig:mug}), and several of its moral subgraphs. \rref{fig:mug6} shows that removing a vertex may have no effect on the removal or fill score. \rref{fig:mug7} shows that removing a vertex can lead to the removal of an edge between its parents.  \rref{fig:mug8} shows that the lack of an edge between parents of a collider in $\cD$ is not sufficient for removal to occur, in this case because they share another child. Finally \rref{fig:mug4} is an example of how removing a vertex with descendants in $\cV$ may cause a significant amount of fill. 
\end{example}

Fill edges are closely related to the edges added to an undirected graph during the course of vertex elimination. In fact, Proposition~\ref{prop:fill-characterization} establishes that the fill edge set for a node is precisely the same as in vertex elimination.
The proof of this proposition is trivial in the case of multivariate Gaussians.
The \textit{marginal precision matrix} (i.e., precision matrix of the marginal distribution) over $\cV\subset [p]$ of a multivariate Gaussian with precision matrix $\Theta_{[p]}$ may be computed via the recursive formula 
$$
(\Theta_{\cV \setminus k})_{ij} = (\Theta_\cV)_{ij} - \frac{1}{\Theta_{kk}} (\Theta_\cV)_{ik} (\Theta_\cV)_{jk}.
$$

In a multivariate Gaussian, the conditional independence statement $X_i \ci X_j \mid \cV \setminus \{ i, j \}$ is equivalent to $(\Theta_{\cV})_{ij} = 0$. Thus, an edge $i - j$ is in the moral subgraph $\cD_\cV(\cD)$ if and only if $(\Theta_\cV)_{ij} = 0$.

Thus, we may conclude that if $(\Theta_{\cV \setminus k})_{ij} \neq 0$ and $(\Theta_\cV)_{ij} = 0$, then $(\Theta_\cV)_{ik} \neq 0$ and $(\Theta_\cV)_{jk} \neq 0$, i.e., if $i - j \in \cR_\cD(\cV, k)$, then $i - k$ and $j - k$.
Proposition~\ref{prop:fill-characterization} establishes the corresponding result in the general, non-parametric case.

\begin{prop}\label{prop:fill-characterization}
The fill edge set of node $k$ over nodes $\cV$ is equal to the neighbors of $k$ in $\cM_\cV(\cD)$ which are not themselves adjacent, i.e.,
\[
\cF_\cD(\cV, k) = \{ i, j \in \nbr_{\cM_\cV(\cD)}(k) \mid i \not\sim_{\cM_\cV(\cD)} j \}.
\]
\end{prop}
\begin{proof}
Let $i,j \in \cV \setminus k$ such that $i - j \in \cM_{\cV \setminus k}(\cD) \setminus \cM_{\cV}(\cD)[\cV \setminus k ]$, i.e. $i - j$ is a fill edge. Then there exists some d-connecting path $\gamma$ from $i$ to $j$ given $\cV \setminus\{i,j,k\}$ that is not d-connecting given $\cV \setminus\{i,j\}$. It follows that (1) $k$ is the only non-collider in $\gamma$ which is also in $\cV \setminus \{i,j\}$, and (2) each collider in $\gamma$ has a descendent in $\cV \setminus \{i,j\}$. Thus, the sub paths $\gamma(i,k)$ and $\gamma(k,j)$ are d-connecting given $\cV \setminus \{i,k\}$  and $\cV \setminus \{k,j\}$ respectively. In other words, $i - k \in \cM_{\cV}(\cD)$,  $j - k \in \cM_{\cV}(\cD)$. By assumption $i \not \sim_{\cM_{\cV}(\cD)} j$, so we conclude that $i - j \in \cM_{\cV \setminus k}(\cD) \setminus \cM_{\cV}(\cD)[\cV \setminus k ] \implies i,j \in \{ i, j \in \nbr_{\cM_\cV(\cD)}(k) \mid i \not\sim_{\cM_\cV(\cD)} j \}$.

Conversely let $i, j \in \cV \setminus \{k\}$ such that $i,j \in \nbr_{\cM_\cV(\cD)}(k) $ and $i \not\sim_{\cM_\cV(\cD)} j$, then there exist two d-connecting paths $\gamma_1$ from $i$ to $k$ given $\cV \setminus \{i,k\}$, and $\gamma_2$ from $j$ to $k$ given $\cV \setminus \{j,k\}$. Since $i \not\sim_{\cM_\cV(\cD)} j$ there is no d-connecting path from $i$ to $j$ given $\cV \setminus \{i,j\}$, so $k$ must be a non-collider in the concatenated path $\gamma_1 \circ \gamma_2$ given $\cV \setminus \{i,j\}$. But if we restrict the conditioning set to $\cV \setminus \{i,j,k\}$, then $\gamma_1 \circ \gamma_2$ becomes d-connecting and as a result $i - j \in \cM_{\cV \setminus k}(\cD)$. Since we assumed $i$ and $j$ are not d-connected given $\cV \setminus \{i,j\}$, it follows that if $i, j \in \nbr_{\cM_\cV(\cD)}(k)$ and $i \not\sim_{\cM_\cV(\cD)} j \}$, then $i - j \in \cM_{\cV \setminus k}(\cD) \setminus \cM_{\cV}(\cD)[\cV \setminus k ]$.
\end{proof}

\textbf{Maximal nodes.} To discover a permutation with a sparse minimal IMAP, we will build the permutation from the last vertex to the first. At each step, we would like to pick a vertex which has no descendants remaining. Formally, if $\cV$ is the set of vertices left unpicked, we seek $k \in \maximal(\cD, \cV)$, where
\[
\maximal(\cD, \cV) := \{ v \in \cV \mid \descendants_\cD(v) \cap \cV = \emptyset \}
\]
is the set of \emph{maximal nodes} of $\cV$ with respect to $\cD$. The next two propositions establish that the removal score and fill score are helpful indicators of whether or not a vertex is maximal.
\vspace{0.2cm}

\begin{prop}\label{prop:removal-score-zero}
$R_\cD(\cV, k) > 0 \Rightarrow k \in \maximal(\cD, \cV)$.
\end{prop}
\vspace{-0.3cm}
\begin{proof}
We prove the contrapositive. Let $k' \in \cV$ be a descendant of $k$ and let $i,j \in \cV \setminus \{k\}$ such that $ i-j \in \cM_\cV(\cD) \left[\cV \setminus \{k\}\right]$, i.e, there is a d-connecting path $\gamma$ from $i$ to $j$ given $\cV \setminus \{i,j\}$ in $\cD$. Assume $k$ is a descendent of a collider in $\gamma$, then $k'$ is as well because $k'$ is a descendant of $k$. As a result, the path will remain d-connecting if we restrict the conditioning set to $\cV \setminus \{i,j,k\}$ because  $k' \in \cV \setminus \{i,j,k\}$. Thus, we conclude that $i - j \in  \cM_{\cV \setminus k}(\cD)$. Note that if $k$ is not a descendent of a collider in $\gamma$, then it cannot be on the path; otherwise, $\gamma$ would not be a d-connecting given $\cV \setminus \{i,j\}$. Thus, in this case $i - j \in  \cM_{\cV \setminus k}(\cD)$ still holds. These results show that $\cM_\cV(\cD)[\cV \setminus k] \subset \cM_{\cV \setminus k}(\cD)$ and it follows that $R_\cD(\cV,k) = 0$. 
\end{proof}

Proposition~\ref{prop:removal-score-zero} gives us a way to \textit{certify} that a node $k$ has no descendants remaining, and thus adding it to the end of the node ordering will be topologically consistent. 
However, the converse is not true: a node may have no descendants remaining, but still have removal score zero. This happens for example if the parents of $k$ form a clique in $\cD$, as in Figure~\ref{fig:mug6}. Thus, a certificate is not \textit{always} available. 
If every node has zero removal score, then we cannot use Proposition~\ref{prop:removal-score-zero} to find maximal nodes. Instead, we resort to proving that some nodes are \textit{not} maximal, which helps prune the search space and increase the likelihood that we pick a maximal node.

\vspace{0.2cm}
\begin{prop}\label{prop:descendants-no-fill}
$F_\cD(\cV, k) > 0 \Rightarrow k \not\in \maximal(\cD, \cV)$.
\end{prop}
\vspace{-0.3cm}
\begin{proof}
We prove the contrapositive. Let $i,j \in V \setminus \{k\}$ such that $i, j \in \nbr_{\cM_\cV(\cD)}(k)$, then there exist two d-connecting paths: $\gamma_1$ from $i$ to $k$  given $\cV \setminus \{i,k\}$,  and $\gamma_2$ from $j$ to $k$ given $\cV \setminus \{j,k\}$. We wish to show that $\gamma := \gamma_1 \circ \gamma_2$ is a d-connecting path from $i$ to $j$ given $\cV \setminus \{i,j\}$, which by Proposition~\ref{prop:fill-characterization} will imply that $F_\cD(\cV, k) = 0$. To this end, we will prove that $k$ must be a collider on $\gamma$. Suppose otherwise, and without loss of generality assume $k$ is of the form $k \rightarrow \dots i$, then we claim that there exists a collider on the path $\gamma_1$. Otherwise, $i$ would be a descendant of $k$ in $\cV$. Thus, let $\ell$ be the collider closest to $k$ in $\gamma_1$ and let $d$ be a descendent of $\ell$ in $\cV$. We know that $d$ exists because if it did not then $\gamma_1$ would not a d-connecting path given $\cV \setminus \{i,k\}$. Note that $d$ is also a descendent of $k$, but since $d \in \cV$ this contradicts the hypothesis. Thus we conclude that $k$ must be a collider, which implies $\gamma$ a d-connecting path from $i$ to $j$ given $\cV \setminus \{i,j\}$.
\end{proof}

However, again, the converse is not true. A non-maximal node $k$ may have zero fill score, for instance if the descendants of $\cV$ form a clique in $\cD$.




\begin{algorithm}[t]
	\caption{\RFD}
	\label{alg:rfd-alg}
	\begin{algorithmic}
		\State \textbf{Input:} Distribution $\bbP$, depth $w$
		\State \textbf{Output:} Permutation $\pi$ 
		\State Estimate $\cM(\bbP)$ via any undirected structure-learning algorithm
		\State Let $\cV_0 = \cV$
		\State Let $\pi = []$
		\While {$|\cV_t| > 0$}
		    \State Pick $\pi' = \RFDStep(\bbP, \cV_t, w)$
		    \State Let $\pi = \pi, \pi'$
		    \State $\cV_t = \cV_{t-1} \setminus \pi'$
		\EndWhile
		\State \textbf{return} $\pi$
	\end{algorithmic}
\end{algorithm}

\begin{algorithm}[t]
	\caption{\RFDStep}
	\label{alg:rfd-step}
	\begin{algorithmic}[1]
		\State \textbf{Input:} Distribution $\bbP$, $\cV_t$, depth $w$
		\State \textbf{Output:} Permutation $\pi$ 
		\State $\paths = [()]$
		\State $r^* = 0$
		\While {$w > 0$ and $r^* = 0$}
		\State $\newpaths = []$
		\For {$\Path \in \paths$}
		\State Let $r_k = R(\cV_t - \Path, k)$, $k \in \cV_t - \Path$
		\If {$\max_k r_k > 0$} 
		    \For {$k \in \arg\max r_k$} \label{line:pick-max-remove}
		        \State $\append~\langle \Path, k \rangle$ to \newpaths
		    \EndFor
		\Else \label{line:else-zero-fill}
    		\State Let $f_k = F(\cV_t - \Path, k)$, $k \in \cV_t - \Path$
    		\For {$k \in \arg\min f_k$} \label{line:pick-min-fill}
		        \State $\append~\langle \Path, k \rangle$ to \newpaths
    		\EndFor
		\EndIf
		\EndFor
		\State $w \gets w-1$
		\State $\paths \gets \newpaths$
		\State $r^* := \max_\Path r_\Path$
	    \EndWhile
	    
	    \State \textbf{return} $\Path \in \arg\min_{\{ \Path \mid r_\Path = r^* \}} d_\Path$ \label{line:return-path}

	\end{algorithmic}
\end{algorithm}

\section{Method}\label{sec:methods}

The above theoretical results suggest using a combination of the removal and fill scores to discover the ordering of the nodes in the graph. In this section, we develop a method based on the removal score, fill score, and the degree or a node; we call this method the \textit{Removal-Fill-Degree} (\RFD) algorithm.

\rref{alg:rfd-alg} begins by estimating an undirected graph over all of the variables. In the multivariate Gaussian case, this can easily be done by thresholding the partial correlation matrix, but this estimator can only be computed if there are more samples than variables ($n > p$) and it has poor  performance if $n\approx p$. Fortunately, there is a large literature on undirected graph estimation in the sparse high-dimensional setting. For example, the CLIME estimator \citep{cai2011constrained}, which estimates a precision matrix as a minimum $\ell$-1 norm estimate under the constraint that its inverse is entry-wise close to the sample covariance matrix, converges in spectral norm to the true precision matrix with rate $s \sqrt{\log p/n}$, where $s$ is the number of nonzero entries in the true precision matrix.

\begin{figure*}
    \includegraphics[width=.9\textwidth]{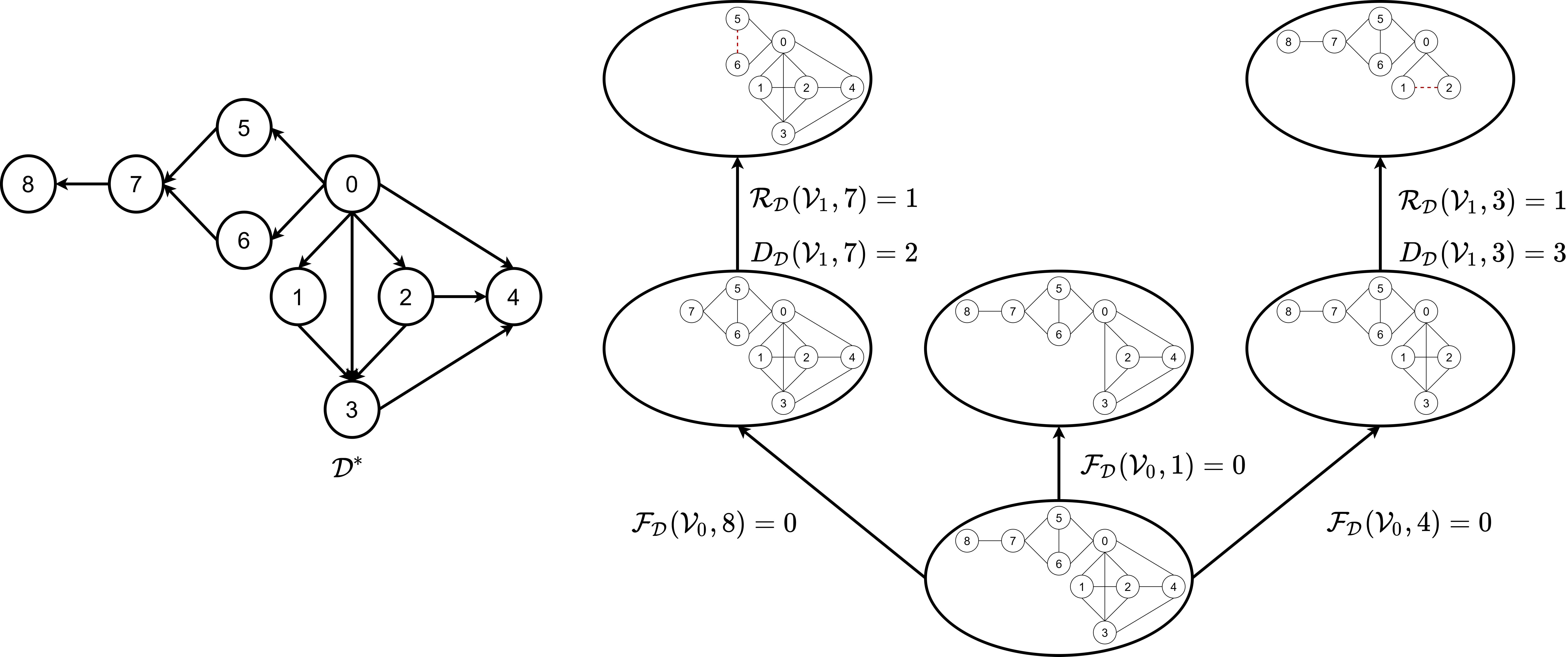}
    \caption{The example call to \RFDStep~described in \rref{ex:rfd-step}. $D_\cD(\cV, k)$ denotes the degree of $k$ in $\cM_\cD(\cV)$.}
    \label{fig:alg-example}
    \vspace{-0.2cm}
\end{figure*}

\textbf{Description of \RFD.} The main principle behind our algorithm is to search for maximal nodes. 
At each step, we perform a breadth-first search of depth $w$ to greedily pick the ``best" set of up to $w$ nodes to add to the end of the order. 
For each $\Path$ in our depth-first search, we define the quantities $r_\Path$ and $d_\Path$ to be the removal and degree scores of the most recently added node to the path. 
As suggested by Proposition~\ref{prop:descendants-no-fill}, a nonzero removal score for a node indicates that it is maximal. 
Thus, if any node has nonzero removal score, we pick amongst the nodes with maximum removal score (\rref{line:pick-max-remove}) and append them to our search path. 
If all nodes have zero removal score, then Proposition~\ref{prop:descendants-no-fill} suggests a way to prune some search directions, since any node with positive fill score is \textit{not} a candidate sink. 
In the noiseless case, we would only need to limit our search to nodes with zero fill score.
However, on real data, we may find that all nodes have positive fill score due to noise, so we pick amongst the nodes with \textit{minimum} fill score (\rref{line:pick-min-fill}) and append them to our search path. 
If any path ends in a node with nonzero removal score ($r^* > 0$), then we exit the search and return one of the paths with maximum removal score, with tie-breaking giving preference to nodes with smaller degree, (\rref{line:return-path}), since this prefers sparser graphs.



%

The following example demonstrates \RFD~ with $w = 2$. 
\begin{example}\label{ex:rfd-step}
Let the true DAG be $\cD^*$ as shown in Figure~\ref{fig:alg-example}. On the first level of the breadth-first search, no node has positive removal score, and the nodes with minimum fill score (in this case, 0 since there is no noise) are 1, 4, and 8 (we take the \texttt{else} branch on \rref{line:else-zero-fill}). On the second level of the breadth-first search, we find that removing 7 after 8, or removing 1 after 4, both result in a removal score of 1, so we add to both paths (\rref{line:pick-max-remove}). Finally, tie-breaking between these two paths is done by picking the path with smaller degree (\rref{line:return-path}), so we add 7,8  to the end of the permutation. 
\end{example}

\subsection{Runtime}

We now characterize the runtime of \RFD, run with depth $w$.
\vspace{0.2cm}

\begin{prop}\label{prop:general-runtime}
    Suppose that updating the undirected graph after marginalization takes $f(n, p)$ time, for $n$ nodes and $p$ samples. Then \RFDStep~with depth $w$ takes $O(p^w (p^2 + f(n, p)))$ time, and \RFD~with depth $w$ takes $O(p^{w+1} (p^2 + f(n, p)))$ time.
\end{prop}
\vspace{-0.3cm}
\begin{proof}
At each step of the breadth-first search in \RFDStep, we need to calculate the removal and fill scores for up to $p^w$ nodes. Given the undirected graph, both of these quantities take at most $O(p^2)$ time to compute. Thus, each \RFDStep~takes $O(p^w (p^2 + f(n,p)))$ time. \RFD~calls \RFDStep~at most $p$ times.
\end{proof}

In the multivariate Gaussian case, when using the partial correlation thresholding estimator, we have $f(n, p) = O(p^2)$, as described in \rref{app:update-gaussian}. This gives us the following corollary:
\begin{cor}\label{cor:gaussian-runtime}
In the multivariate Gaussian setting, \RFDStep~with depth $w$ takes $O(p^{w+2})$ time, and \RFD~with depth $w$ takes $O(p^{w+3})$ time. In particular, \RFD~with depth 1 takes $O(p^4)$ time.
\end{cor}

In comparison, most \textit{provably consistent} causal structure learning algorithms require bounds on certain graph parameters, such as maximum indegree, in order to achieve a polynomial run time. For instance, the prominent PC algorithm must perform $O(p^{k+2})$ conditional independence tests, where $k$ is the maximum indegree of the true DAG. Similarly, recent versions of GES \citep{chickering2015selective} require $O(p^{k+2})$ calls to a scoring function. In the case of GSP~\cite{solus2017consistency}, the complexity at each step depends on the size of the Markov equivalence class, rather than the usual scaling based on maximum indegree.

\textbf{Performance on Dense Graphs.}
Since \RFD~is able to run in polynomial time without an explicit sparsity assumption on the underlying graph, and \RFD~is an approximate algorithm, a natural question arises: ``does \RFD~always perform poorly when the underlying graph is not sparse?" 
We show that the answer to this question is ``no": there exist dense graphs on which our method performs well; i.e., the combination of speed and performance achieved by our method does not rely on sparsity of the underlying graph.

Let $B_K$ denote a graph on $K + {K \choose 2}$ nodes, with edges generated as follows. For each $j \in K+1, \ldots, K + {K \choose 2}$, pick $P_j \subseteq [K]$, with $|P_j| = 2$, such that $P_j \neq P_{j'}$ for any $j \neq j'$. Let $p \to j$ for $p \in P_j$, and let there be a complete graph on $K_1, \ldots, K + {K \choose 2}$, with topological order given by numerical order. \rref{fig:dense-graph} shows an example of this construction for $K = 4$.

\begin{figure}
    \centering
    \includegraphics[width=.7\linewidth]{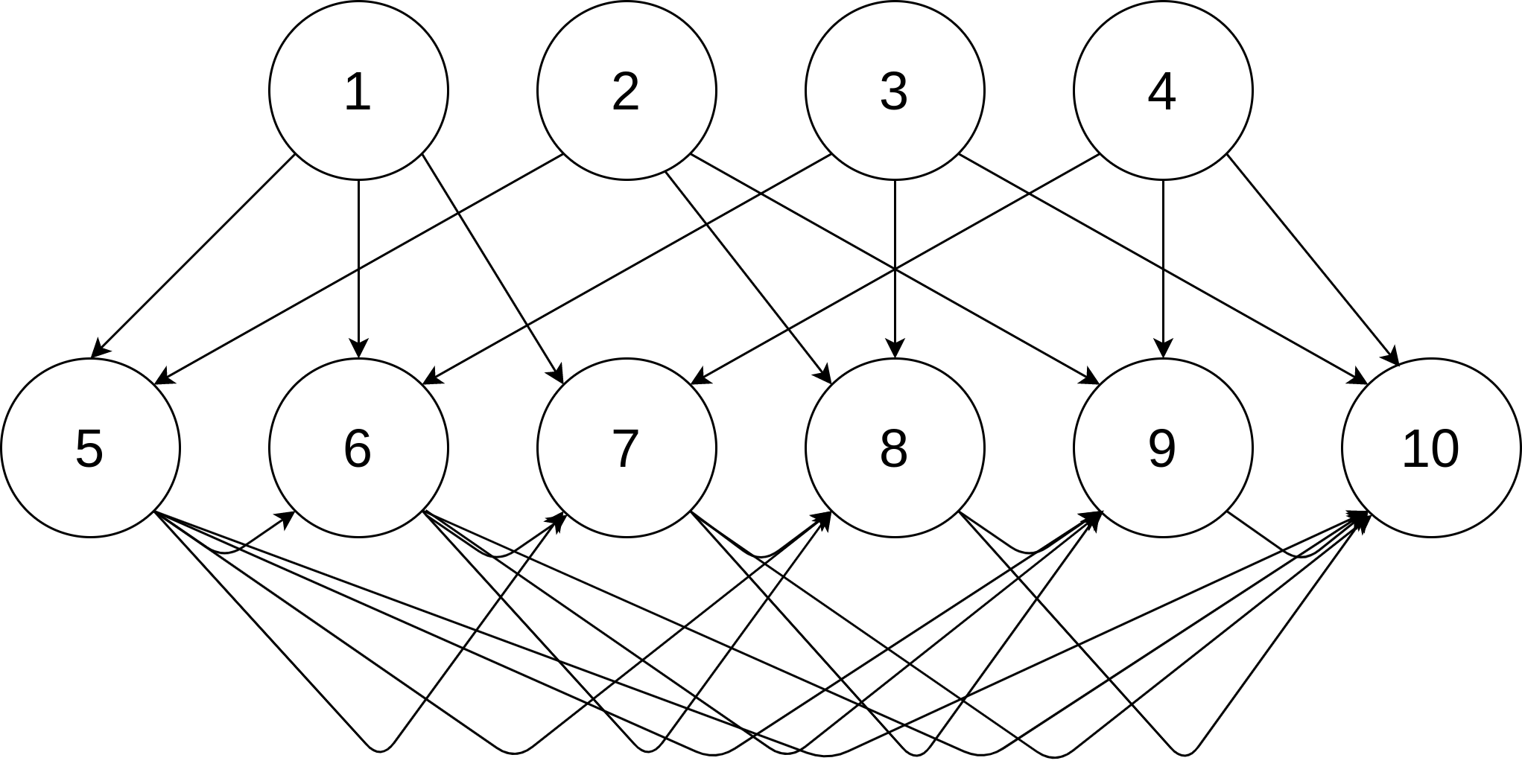}
    \vspace{-0.2cm}
    \caption{$B_4$}
    \label{fig:dense-graph}
    \vspace{-0.2cm}
\end{figure}

The number of missing edges in $B_K$ is less than ${K \choose 2} + K{K \choose 2} = O(K^3)$, whereas the number of \textit{possible} edges is ${{K \choose 2} + K \choose 2} = O(K^4)$, so that $B_K$ is dense.
Furthermore, the \RFD~algorithm perfectly recovers $B_K$, since the removal score of each $j \in K+1, \ldots, K + {K \choose 2}$ becomes exactly 1 only after all nodes after in the ordering are removed, and the ordering of the first $K$ nodes is arbitrary.



\begin{figure*}[t]
\vspace{-0.2cm}
    \begin{subfigure}{.33\textwidth}
    \centering
    \includegraphics[width=\textwidth]{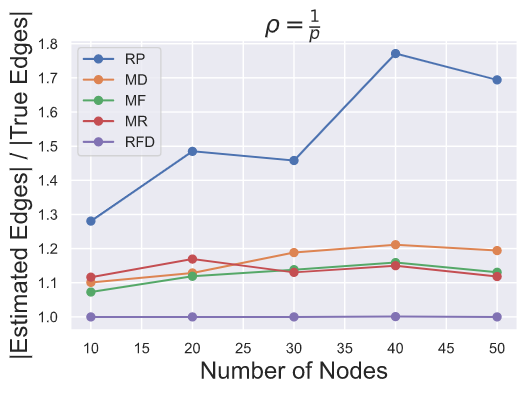}
    \label{fig:noiseless-1-over-p}
    \end{subfigure}
    ~
    \begin{subfigure}{.33\textwidth}
    \centering
    \includegraphics[width=\textwidth]{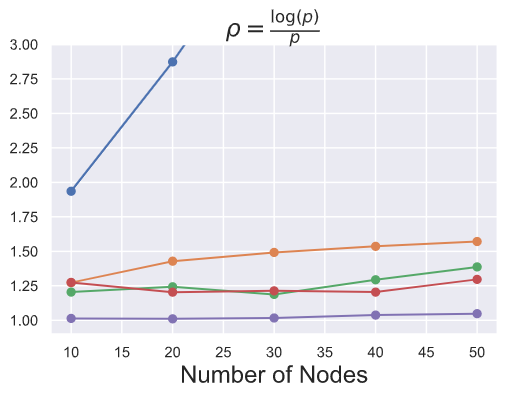}
    \label{fig:noiseless-1-over-p}
    \end{subfigure}
    ~
    \begin{subfigure}{.33\textwidth}
    \centering
    \includegraphics[width=\textwidth]{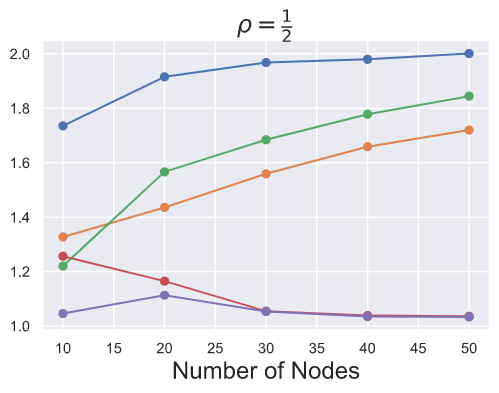}
    \label{fig:noiseless-1-over-p}
    \end{subfigure}
    \vspace{-0.6cm}
    \caption{
    Performance of permutation-finding algorithms in the noiseless setting, measured by the density of their induced minimal IMAPs relative to the true graph. Each point represents the average over 100 randomly generated DAGs.
    }\label{fig:noiseless-results}
    \vspace{-0.2cm}
\end{figure*}

\section{Empirical Results}\label{sec:empirical}

In this section, we generate DAGs according to an Erd{\"o}s-R{\'e}nyi skeleton with random order, varying the number of nodes $p$ and the edge density $\rho$, which may be a function of $p$. We pick edge weights independently from $\textrm{Unif}([-1, -.25] \cup [.25, 1])$ to ensure that they are bounded away from zero.

\subsection{Noiseless Setting}
\vspace{-0.2cm}

We first investigate the quality of the permutation found by our algorithm in the noiseless setting, i.e., when we are given the true precision matrix. 
Given the output permutation $\pi$ of \RFD, we may return the graph estimate $\cD_\pi$, i.e., the minimal IMAP discussed in \rref{sec:background}. The performance of \RFD~can be measured by the ratio of the number of edges in $\cD_\pi$ to the number of edges in the true graph, $\cD^*$. 
This ratio is always greater than or equal to 1, with equality if and only if $\cD_\pi$ is in the Markov equivalence class of $\cD^*$.

We compare to a number of baselines, including \textbf{random permutations (RP)} and the following greedy selection strategies, where $\cV_t$ is the set of unpicked nodes at step $t$ of the algorithm and $k_t$ is the node picked at step $t$:
\vspace{-0.2cm}
\begin{itemize}
    \item \textbf{Min-degree (MD)}: $k_t \in \arg\min_{k} |\nbr_{\cM_{\cV_t}(\cD)}(k)|$
    \vspace{-0.6cm}
    \item \textbf{Min-fill (MF)}: $k_t \in \arg\min_{k} F_\cD(\cV_t, k)$
    \vspace{-0.2cm}
    \item \textbf{Max-remove (MR)}: $k_t \in \arg\max_{k} R_\cD(\cV_t, k)$
\end{itemize}
\vspace{-0.2cm}

Figure~\ref{fig:noiseless-results} demonstrates that the \RFD~algorithm is often able to find a permutation which induces a minimal IMAP that is nearly as sparse as the true DAG. The \RFD~algorithm clearly outperforms all of the baselines on this task. 
It is notable that on dense graphs ($\rho = \frac{1}{2}$), with a large number of nodes ($p \geq 30$), the MR algorithm matches the performance of the \RFD~algorithm, indicating that the removal score is an especially valuable way to identify maximal nodes in such settings.
In contrast, on dense graphs, the performance of the MD and MF algorithms both \textit{degrade} as the number of nodes increases. 

\begin{figure*}[t!]
\vspace{-0.2cm}
    \centering
    \includegraphics[width=\textwidth]{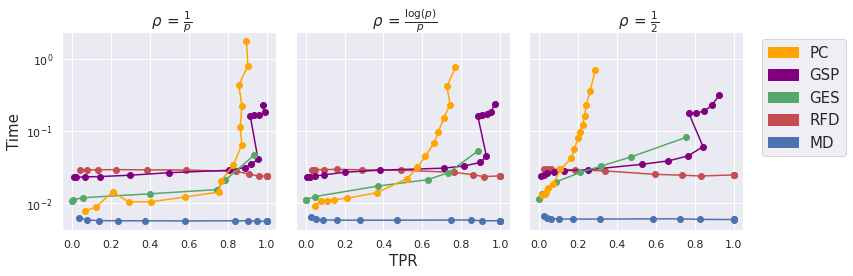}
    \vspace{-0.8cm}
    \caption{
    Computation time of various causal structure learning algorithms, with $p = 20$ nodes and $n = 40$ samples. Each point represents the average over 35 randomly generated DAGs.
    }\label{fig:noisy-results-time}
    \vspace{-0.4cm}
\end{figure*}

\begin{figure*}[t!]
    \begin{subfigure}{.33\textwidth}
    \centering
    \includegraphics[width=\textwidth]{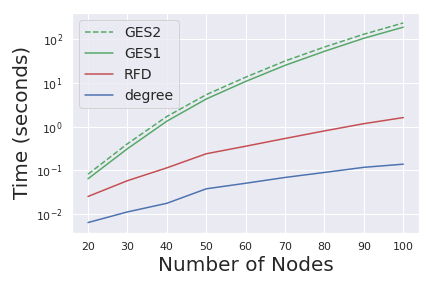}
    \label{fig:time-scaling}
    \end{subfigure}
    ~
    \begin{subfigure}{.33\textwidth}
    \centering
    \includegraphics[width=\textwidth]{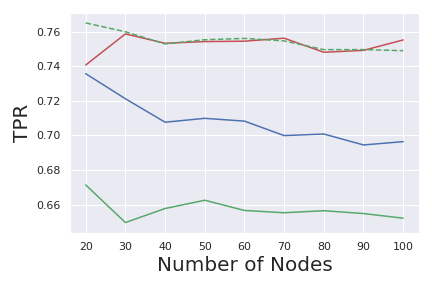}
    \label{fig:tpr-scaling}
    \end{subfigure}
    ~
    \begin{subfigure}{.33\textwidth}
    \centering
    \includegraphics[width=\textwidth]{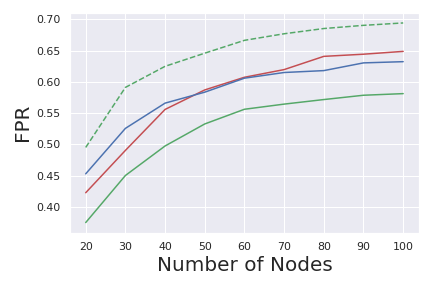}
    \label{fig:fpr-scaling}
    \end{subfigure}
    \vspace{-0.6cm}
    \caption{
    Performance and computation time of GES, \RFD, and MD as a function of the number of nodes. Given $p$ nodes, we take $n = 20p$ samples.
    }\label{fig:scaling}
    \vspace{-0.3cm}
\end{figure*}

\subsection{Noisy Setting}
\vspace{-0.2cm}

We now compare the performance of \RFD~to that of other algorithms on the task of causal structure learning from data.
We first investigate the computational scaling of our method along with several other prominent methods for causal structure learning across a range of densities. In \rref{app:roc-noisy}, we show that the ROC curves of all of the algorithms are similar across the range of densities. 
Since all algorithms perform similarly, we focus on how their computation time grows as a function of how dense of a graph they estimate, measured via the true positive rate.
As evidenced by Figure~\ref{fig:noisy-results-time}, the computation times of PC, GSP, and to a lesser extent GES and GSP all scale poorly as their sparsity parameters are tuned to yield denser graph.
Meanwhile, the computation times of \RFD~ and MD are almost constant across their ranges.
These results suggest that, especially in dense and/or high-dimensional regimes, the \RFD~algorithm can be used as a computationally efficient alternative to existing causal structure learning algorithms.

In \rref{fig:scaling}, we further compare \RFD~and MD to GES as a function of the number of nodes in the graph, in the dense setting of $\rho = \frac{1}{2}$. For all algorithms, we aimed to pick their parameter such that the true positive rate was approximately .7 on 20 node graphs, since this was uniformly a point on the ROC curve which offered a good compromise between true positive rate and false positive rate.
In the case of \RFD~and MD, we used a significance level of $\alpha = .001$ for all hypothesis tests. For GES, we picked two values of the regularization parameter $\lambda$ which ``sandwich" the TPR and FPR of our algorithms; $\lambda_1 = .5$ (in GES1) and $\lambda_2 = 1$ (in GES2). 
We find that while the performance of GES is comparable to that of \RFD, the required computation time scales much more dramatically. Moreover, \RFD~beats the slightly faster MD algorithm at almost every point, with a higher true positive rate and lower false positive rate.

\section{Discussion}\label{sec:discussion}

In this paper, we introduced a novel, efficient method for approximately recovering the node ordering of a causal DAG. 
Our method, the \RFD~algorithm, is motivated by the minimum-degree algorithm for sparse Cholesky decomposition, but leverages additional DAG-specific structure for improved performance. 
In particular, our method is based on the phenomenon of edge \textit{removal} after marginalization of a sink node.
Our method systematically combines signals about the causal ordering from a combination of edge removal, edge addition, and node degrees, and greatly outperforms methods which use only a single one of these signals on the task of permutation discovery. 
Moreover, running our method with fixed depth offers a polynomial-time alternative to provably consistent causal structure learning algorithms, which only run in polynomial time under assumptions on the underlying graph. 
We demonstrate that using our algorithm for causal structure learning performs comparably to existing causal structure learning algorithms, but with a significant speedup in run time.

Developing scalable causal structure learning algorithms is critical, since many of the domains in which causal structure learning is valuable involve thousands to millions of variables, for example in genomics \citep{bucur2019large,belyaeva2020causal} and neuroscience \citep{dubois2017causal}. Since the \RFD~algorithm is not specific to using the partial correlation thresholding estimator for the undirected graph, our algorithm can be used in non-Gaussian and high-dimensional settings. 
In the non-Gaussian case, our algorithm can be combined with nonparametric conditional independence tests such as HSIC \citep{gretton2008kernel} for estimating the undirected graph. 
In the high-dimensional case, there are a variety of estimators with guarantees for sparse graphs. 
One current advantage of the partial correlation thresholding estimator, in comparison to these estimators, is the ability to update the moral subgraph via a rank-one matrix addition in $O(p^2)$ time. 
To put other estimators into practice with our algorithm, especially on large graphs, it is necessary to develop efficient ways of updating estimates after marginalization, with provable guarantees that such updates do not introduce additional error.

\section*{Acknowledgments}
Chandler Squires was partially supported by an NSF Graduate Fellowship, MIT J-Clinic for Machine Learning and Health, and IBM. Caroline Uhler was partially supported by NSF (DMS-1651995), ONR (N00014-17-1-2147 and N00014-18-1-2765), and a Simons Investigator Award.

\bibliography{bib}

\clearpage
\appendix
\newcommand{\snum}{S}
\renewcommand{\theequation}{\snum.\arabic{equation}}

{\Large\textbf{Supplementary Material}}

\section{Efficiently updating the undirected graph for multivariate Gaussians}\label{app:update-gaussian}

We first describe the partial correlation thresholding estimator of the moral subgraph, showing that it takes $O(p^2)$ time given the sample precision matrix $\hatTheta^\cV$. Then, we show that after marginalizing a node $k$, the sample precision matrix $\hatTheta^{\cV \setminus k}$ takes $O(p^2)$ time to compute. Thus, by retaining the sample precision matrix over the current set of nodes at each iteration of the \RFD~algorithm, we may compute the new undirected graph in $O(p^2)$ time.

\subsection{Partial correlation thresholding estimator}

The \emph{partial correlation} between $X_i$ and $X_j$ given $X_S$, denoted $\rho_{i,j\mid S}$, is equal to the correlation of the residuals of $X_i$ and $X_j$ after performing linear regression on $X_S$. Supposing $X$ has a multivariate Gaussian distribution, recall that
$$
\rho_{ij\mid S} = 0 \Longleftrightarrow X_i \ci X_j \mid X_S.
$$
A classical result states that if $\rho_{i,j\mid S} = 0$, and $\hatrho_{ij\mid S}$ is the sample partial correlation computed from $n$ samples, then the quantity
\begin{equation}
    \hatz_{ij\mid S} = \sqrt{n - |S| - 3} \left|\frac{1}{2} \log \left( \frac{1 + \hatrho_{ij\mid S}}{1 - \hatrho_{ij\mid S}} \right)\right| \label{eq:z-transform}
\end{equation}
is distributed as a standard normal; i.e., to test the null hypothesis $H_0: \rho_{ij\mid S} = 0$ at significance level $\alpha$, we can reject if $|\hat{z}_{ij \mid S}| \geq \Phi^\inv(1 - \alpha/2)$.

Let $\hatTheta^\cV$ denote the marginal sample precision matrix over $X_\cV$, i.e., $\hatTheta_\cV = \hatSigma_{\cV,\cV}^\inv$, where $\hatSigma$ is the sample covariance matrix. Then the matrix of sample partial correlations $\hatK^\cV= [\hatrho_{ij \mid \cV \setminus \{ i, j \}}]_{ij}$ can be efficiently computed from $\hatTheta^\cV$ via the following formula:
$$
\hatK^\cV_{ij} = -\frac{\hatTheta^\cV_{ij}}{\sqrt{\hatTheta^\cV_{ii} \hatTheta^\cV_{jj}}}.
$$

Applying \rref{eq:z-transform} element-wise to $\hatK^\cV$ and thresholding gives an estimate of the moral subgraph $\cM_\cV(\cD)$.

We perform $O(1)$ operations on each element of $\hatTheta_\cV$, so computing the moral subgraph given $\hatTheta_\cV$ is $O(p^2)$.

\subsection{Updating the Marginal Precision Matrix}

If we consider the effect of marginalizing out $k$, the new marginal sample precision matrix is related to $\hatTheta^S$ by the following rank-one update:
$$
\hatTheta^{\cV \setminus k} = \hatTheta^\cV_{\cV \setminus k, \cV \setminus k} - (\hatTheta^\cV_{kk})^\inv \hatTheta^\cV_{\cV \setminus k, k} \hatTheta^S_{k, \cV \setminus k}.
$$
Thus, given access to $\hatTheta^\cV \in \bbR^{p \times p}$, we may compute $\hatTheta^{\cV \setminus k}$ in $O(p^2)$ time.

\section{Performance on 20-node graphs}\label{app:roc-noisy}

\rref{fig:noisy-results} shows that the various causal structure learning algorithms which we test perform similarly.

\begin{figure*}
    \centering
    \includegraphics[width=\textwidth]{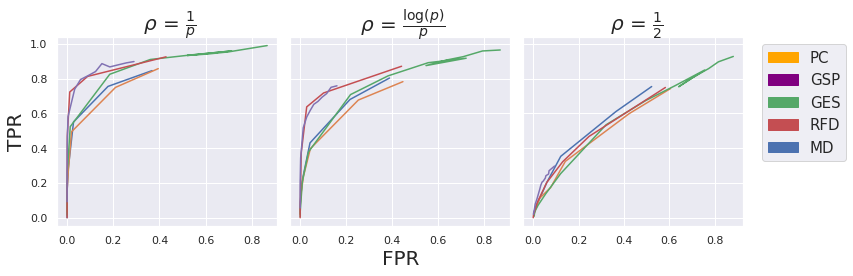}
    \caption{ROC curve}
    \label{fig:noisy-results}
\end{figure*}

\end{document}